\documentclass[letterpaper, 10 pt, conference]{ieeeconf}  % Comment this line out if you need a4paper

\newif\ifdraft
\draftfalse

\newcommand{\change}[1]{\ifdraft{\color{blue}#1}\else{#1}\fi}

\usepackage{blindtext}
\usepackage{graphicx}
\usepackage[utf8]{inputenc}
\usepackage{graphics} % for pdf, bitmapped graphics files
\usepackage{epsfig} % for postscript graphics files
\usepackage{mathptmx}
\usepackage{amsmath} % assumes amsmath package installed
\usepackage{amssymb}  % assumes amsmath package installed
\usepackage{cite}
\usepackage{multicol}
\usepackage{mathtools, cuted}
\usepackage[cmintegrals]{newtxmath}
\usepackage{epstopdf}
\usepackage{graphics} % for pdf, bitmapped graphics files
\usepackage{epsfig} % for postscript graphics files
\usepackage{mathptmx} % assumes new font selection scheme installed
\usepackage{amsmath} % assumes amsmath package installed
\usepackage{amssymb}  % assumes amsmath package installed
\usepackage{cite}
\usepackage{multicol}
\usepackage{mathtools, cuted}
\usepackage[cmintegrals]{newtxmath}
\usepackage{mathtools}

\usepackage{mathptmx}
\usepackage{helvet}
\usepackage{courier}
\usepackage{subfigure}
\usepackage{type1cm}
\usepackage{epsfig} % for postscript graphics files
\usepackage{mathptmx} % assumes new font selection scheme installed
\usepackage{amsmath} % assumes amsmath package installed
\usepackage{amssymb}  % assumes amsmath package installed
\usepackage{cite}
\usepackage{multicol}
\usepackage{mathtools, cuted}
\usepackage[cmintegrals]{newtxmath}
\usepackage{makeidx}         % allows index generation
\usepackage{algorithmic}

\newtheorem{theorem}{Theorem}
\newtheorem{proposition}{Proposition}
\usepackage{multicol}        % used for the two-column index
\usepackage[bottom]{footmisc}% places footnotes at page bottom
\usepackage{caption}
\usepackage{nomencl}
\usepackage[usenames, dvipsnames]{color}

\IEEEoverridecommandlockouts                              % This command is only needed if 
\title{\LARGE \bf
Formal Specification of
Continuum Deformation Coordination
}

\author{Hossein Rastgoftar, Jean-Baptiste Jeannin, and Ella Atkins
\thanks{Authors are with the Aerospace Engineering Department, University of Michigan, Ann Arbor,
MI, 48109 USA e-mails: \{hosseinr, jeannin,ematkins\}@umich.edu }}

\renewcommand\dim{{\ensuremath d}}
\newcommand\htr[1]{#1^{\rm HT}}

\begin{document}

\maketitle
\thispagestyle{empty}
\pagestyle{empty}

%%%%%%%%%%%%%%%%%%%%%%%%%%%%%%%%%%%%%%%%%%%%%%%%%%%%%%%%%%%%%%%%%%%%%%%%%%%%%%%%
\begin{abstract}
Continuum deformation is a leader-follower multi-agent cooperative control approach. Previous work showed a desired continuum deformation can be uniquely defined based on trajectories of $\dim+1$ leaders in a $\dim$-dimensional motion space and acquired by followers through local inter-agent communication. This paper formally specifies continuum deformation coordination in an obstacle-laden environment. Using linear temporal logic (LTL), continuum deformation liveness and safety requirements are defined.
Safety is prescribed by providing conditions on (i)~agent deviation bound, (ii)~inter-agent collision avoidance, (iii) agent containment, (iv)~motion space containment, and (v)~obstacle collision avoidance. 
Liveness specifies a reachability condition on the desired final formation.
%reachability condition in a continuum deformation coordination.
% Liveness specifies possible continuum deformations given a desired final formation and obstacle geometries in the motion space.

\end{abstract}
\section{Introduction}
%Not needed. - \subsection{Background}
%
From package delivery and autonomous taxis to military applications, Unmanned Aerial Vehicles (UAV) are changing our daily lives.
Some applications however cannot be achieved by a single UAV, but need a swarm of cooperating UAVs forming a Multi-Agent System (MAS).
Examples of such applications are surveillance, formation flight, and traffic control.
% Inter-agent and obstacle collision avoidance and possibility of passing are important coordination issues. 
MAS perform critical tasks, and it is becoming increasingly important to formally specify and verify the correctness of their behavior, in terms of both safety and liveness requirements.
In this paper we are primarily interested in formation flying. We treat MAS evolution as a continuum deformation~\cite{rastgoftar2016continuum}, and formally specify its safety and liveness requirements.

%Multi-agent system (MAS) coordination has been widely studied in the literature.
Multi-agent system coordination
applies methods such as consensus \cite{liu2014asynchronously,  bidram2013secondary} %and partial differential equation (PDE) \cite{kim2008pde, frihauf2011leader, magar2014adaptive}. 
with application to distributed motion control \cite{ren2007information, yu2016distributed}, sensing \cite{li2015distributed, zhang2016distributed}, medical systems \cite{seeff2002national}, and smart grids \cite{zhao2016fully, xing2015distributed}.  
For containment control \cite{cao2011distributed, yoo2015distributed} multiple leaders guide the MAS toward a target shape using consensus to update positions \cite{cao2009containment, cao2011distributed}  under fixed and switching communication topologies  \cite{zhang2017event, li2015containment}. %Containment control modeled with double integrator vehicle dynamics was presented in \cite{lin2013leader}. %Containment control with complex communication weights was investigated in \cite{lin2013leader}. 
Directed communication topologies \cite{mourad2014containment, xu2015containment}, event-based containment control \cite{liu2016event, zhang2017event}, and finite-time containment control  \cite{zhao2015finite} have been formulated.
Formal specification and verification of multi-agent systems \change{have} received considerable attention~\cite{tomlin1998conflict,brazier1995formal,jonker2002compositional,raimondi2007automatic}, and our aim is to extend that work to the context of continuum deformation.
Containment control assures asymptotic convergence to a desired configuration inside the convex region prescribed by leaders but has two limitations: (i) followers are not assured to remain inside the moving convex region defined by leader positions during transition; and (ii) inter-agent collision avoidance cannot be guaranteed for an arbitrary initial agent distribution. Continuum deformation extends containment control theory by prescribing a homogeneous mapping that guarantees inter-agent collision avoidance and that followers remain within the leader-defined boundary \cite{rastgoftar2016continuum, rastgoftar2018continuum}. \change{In a continuum deformation coordination, inter-agent distances can aggressively change while no two particles collide. This property can advance swarm coordination maneuverability and agility, and allows a large-scale MAS to safely negotiate narrow channels in obstacle-laden environments.}

% \subsection{Contribution and Outlines}
\change{As its main contribution, this paper formally specifies} safety and liveness for the coordination of continuum deformation of an MAS with a \change{large} number of agents (Fig.~\ref{schematic}). Using triangulation and tetrahedralization, safety conditions are defined to assure obstacle collision avoidance, inter-agent collision avoidance, and motion space containment in 2-dimensional and 3-dimensional continuum deformations. This paper also formally specifies a liveness condition that assures continuum deformation is possible given an initial MAS configuration and a motion space obstacle geometry. 

\begin{figure}
\center%{\epsfig{figure=fig1.eps,width=6.85in}}
\includegraphics[width=3.3 in]{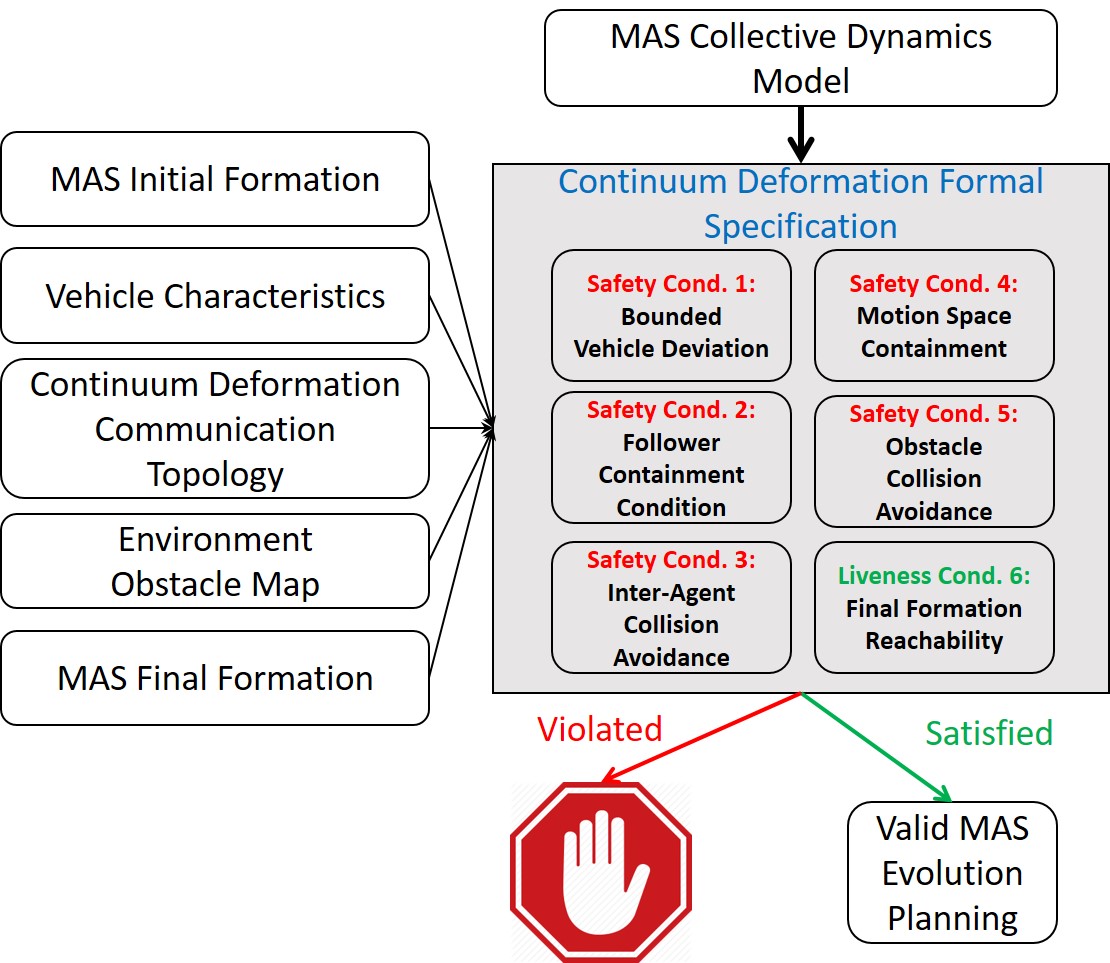}
\caption{Elements of the formal specification.}
\vspace{-20pt}
\label{schematic}
\end{figure}

This paper is organized as follows: In Section \ref{prelim}, preliminaries in triangulation and tetrahedralization, \change{continuum deformation coordination}, graph theory, linear temporal logic, and MAS collective dynamics are reviewed. Continuum deformation formal specification in Section  \ref{problemstatement} is followed by sufficient safety conditions in  Section \ref{Formal Specification}.
% safety and liveness conditions developed in Sections \ref{MAS Continuum Deformation Safety Requirements} and \ref{MAS Continuum Deformation Liveness Condition}, respectively. 
% MAS continuum deformation  coordination is formally specified in Section \ref{Formal Specification}. 
Simulation results and conclusions are presented in Sections \ref{Simulation Results} and \ref{Conclusion}, respectively.
\section{Preliminary Notions}
\label{prelim}

\subsection{Triangulation and Tetrahedralization}
\label{sec:triangulation}
To determine whether an agent stays in its designated motion space and does not collide with any obstacle, we need to compute whether this agent is inside or outside a given $\dim$-dimensional polytope. Our approach creates a partition of the polytope into a number of $\dim$-simplexes (i.e., a triangle for $d=2$ or a tetrahedron for $d=3$), thereby reducing the problem to checking whether our agent stays in one of the simplexes.
% In this section, we explore how to determine whether a given agent at position $\mathrm c\in\mathbb{R}^d$ stays in a $d$-simplex $\mathcal{T}$ defined by points $\mathbf{a}_1,\ldots,\mathbf{a}_{d+1}\in\mathbb{R}^d$.
\change{A $d$-simplex $\mathcal{T}$ is defined as the non-zero volume specified by points $\mathbf{a}_1,\ldots,\mathbf{a}_{d+1}\in\mathbb{R}^d$. Note that $\mathbf{a}_1,\ldots,\mathbf{a}_{d+1}\in\mathbb{R}^d$ form a valid $d$-simplex if \change{and only if} the following rank condition is satisfied:}
%Let $\mathbf{\Lambda}:\underbrace{\mathbb{R}^{\dim}\times\cdots\times\mathbb{R}^{\dim}}\limits_{\dim+1~\mathrm{times}}\rightarrow\{0,\cdots,\dim\}$ be defined as follows:
\begin{equation}
\label{dsimplex}
    \mathbf{\Lambda}\left(\mathbf{a}_1,\cdots,\mathbf{a}_{\dim+1}\right)=\mathrm{rank}\left(\begin{bmatrix}
\mathbf{a}_2-\mathbf{a}_1&\cdots&\mathbf{a}_{\dim+1}-\mathbf{a}_1
\end{bmatrix}\right)=d,
\end{equation}
%
%where $\mathbf{a}_1\in \mathbb{R}^{\dim\times 1}$, $\cdots$, and $\mathbf{a}_{\dim}\in\mathbb{R}^{\dim\times 1}$ ($\dim=2,3$). Assume $\mathbf{\Lambda}\left(\mathbf{a}_1,\cdots,\mathbf{a}_{\dim+1}\right)=\dim$. Then, $\mathbf{a}_1$, $\cdots$, $\mathbf{a}_{\dim+1}$ define vertices of an $\dim$-dimensional simplex in $\mathbb{R}^\dim$, and an arbitrary vector $\mathbf{c}\in \mathbb{R}^\dim$ can be uniquely expressed as follows:
% \[
% \begin{split}
%     \mathbf{c}=&\sum_{i=2}^{\dim+1}\alpha_i\left(\mathbf{a}_i-\mathbf{a}_1\right)=\left(1-\sum_{i=2}^{\dim+1}\alpha_i\right)\mathbf{a}_1+\sum_{i=2}^{\dim+1}\alpha_i\mathbf{a}_i\\
% \end{split}
% \]
% Defining $\alpha_{1}=\left(1-\sum_{i=2}^{\dim+1}\alpha_i\right)$, $\mathbf{c}$ can be expressed as follows: 
% \begin{equation}\label{convexpresss1}
%     \mathbf{c}=\sum_{i=1}^{\dim+1}\alpha_i\mathbf{a}_i,
% \end{equation}
% where 
% \begin{equation}\label{convexpresss2}
%     \sum_{i=1}^{\dim+1}\alpha_i=1.
% \end{equation}
% Parameters $\alpha_1$ through $\alpha_\dim$ are called \emph{$\alpha$-parameters}. Define
% \[
% \mathbf\Omega\left(\mathbf{a}_1,\cdots,\mathbf{a}_{\dim+1},\mathbf{c}\right)=
% \begin{bmatrix}
% \alpha_{1}\\\vdots\\\alpha_{\dim+1}
% \end{bmatrix}
% \in \mathbb{R}^{\left(\dim+1\right)\times 1}.
% \]
\change{If \eqref{dsimplex} is satisfied, we can define vector operator $\Omega$ given an arbitrary vector $\mathbf{c}$ and $\mathbf{a}_1,\ldots,\mathbf{a}_{d+1}$:}
%Considering Eqs. \eqref{convexpresss1} and \eqref{convexpresss2}, 
\begin{equation}
    \mathbf\Omega\left(\mathbf{a}_1,\cdots,\mathbf{a}_{\dim+1},\mathbf{c}\right)=
    \begin{bmatrix}
    \mathbf{a}_1&\cdots&\mathbf{a}_{\dim+1}\\
    1&\cdots&1\\
    \end{bmatrix}
    ^{-1}
    \begin{bmatrix}
    \mathbf{c}\\
    1
    \end{bmatrix}
    .
    \label{eq:omega}
\end{equation}
% is well-defined. 
Note that  $\Omega\left(\mathbf{a}_1,\cdots,\mathbf{a}_{\dim+1},\mathbf{c}\right)\in\mathbb{R}^{d+1}$, let
\[
\begin{bmatrix}
\alpha_{1}\\\vdots\\\alpha_{\dim+1}
\end{bmatrix}=
\mathbf\Omega\left(\mathbf{a}_1,\cdots,\mathbf{a}_{\dim+1},\mathbf{c}\right).
%\in \mathbb{R}^{\left(\dim+1\right)\times 1}.
\]
%$\alpha_1,\ldots,\alpha_{d+1}$ be its coordinates such that:
%Note that $\mathbf{\Omega}:\underbrace{\mathbb{R}^{\dim}\times\cdots\times\mathbb{R}^{\dim}}\limits_{\dim+2~\mathrm{times}}\rightarrow\mathbb{R}^{\dim+1}$ exists, if $\mathbf{\Lambda}\left(\mathbf{a}_1,\cdots,\mathbf{a}_{\dim+1}\right)=\dim$. $\mathbf{\Omega}$ is a one-sum column vector, e.g. the sum of the entries of vector $\mathbf{\Omega}$ is one.
%Because $\mathbf{\Omega}$ is one-sum column, parameters $\alpha_1$ through $\alpha_{\dim+1}$ cannot be all negative.
As shown in Fig. \ref{2Ddisc}, a 2-dimension motion space ($\dim=2$) can be divided into $10$ regions based on the signs of $\alpha_1$, $\alpha_2$, and $\alpha_3$.
% : $\mathrm{zone~1:}\alpha_1=0$, $\mathrm{zone~2:}\alpha_2=0$, $\mathrm{zone~3:}\alpha_3=0$, $\mathrm{zone~4:}\alpha_1<0,~\alpha_2>0,~\alpha_3>0$, $    \mathrm{zone~5:}\alpha_1<0,~\alpha_2>0,~\alpha_3<0$, $   \mathrm{zone~6:}\alpha_1<0,~\alpha_2<0,~\alpha_3>0$, $    \mathrm{zone~7:}\alpha_1>0,~\alpha_2<0,~\alpha_3>0$, $    \mathrm{zone~8:}\alpha_1>0,~\alpha_2<0,~\alpha_3<0$, $    \mathrm{zone~9:}\alpha_1>0,~\alpha_2>0,~\alpha_3<0$, and $    \mathrm{zone~10:}\alpha_1>0,~\alpha_2>0,~\alpha_3>0$.
% \[
% \begin{split}
%     \mathrm{zone~1:}&\alpha_1=0,\\
%     \mathrm{zone~2:}&\alpha_2=0,\\
%     \mathrm{zone~3:}&\alpha_3=0,\\
%     \mathrm{zone~4:}&\alpha_1<0,~\alpha_2>0,~\alpha_3>0,\\
%     \mathrm{zone~5:}&\alpha_1<0,~\alpha_2>0,~\alpha_3<0,\\
%     \mathrm{zone~6:}&\alpha_1<0,~\alpha_2<0,~\alpha_3>0,\\
%     \mathrm{zone~7:}&\alpha_1>0,~\alpha_2<0,~\alpha_3>0,\\
%     \mathrm{zone~8:}&\alpha_1>0,~\alpha_2<0,~\alpha_3<0,\\
%     \mathrm{zone~9:}&\alpha_1>0,~\alpha_2>0,~\alpha_3<0,\\
%     \mathrm{zone~10:}&\alpha_1>0,~\alpha_2>0,~\alpha_3>0.\\
% \end{split}
% \]
Similarly, a 3-dimension motion space can be divided into $55$ regions based on the signs of $\alpha_1$, $\alpha_2$, $\alpha_3$ and $\alpha_4$. In general, we can decide whether $\mathbf{c}$ is inside or outside a simplex based on the signs of $\alpha_1,\ldots,\alpha_{d+1}$:
% \[
% \begin{split}
% \mathrm{zones~1-10:}&\alpha_1=0,~\alpha_2,~\alpha_3,~\alpha_4\in\{\mathbb{R}_-,0,\mathbb{R}_+\},~\\
% \mathrm{zones~11-20:}&\alpha_2=0,~~\alpha_1,~\alpha_3,~\alpha_4\in\{\mathbb{R}_-,0,\mathbb{R}_+\},~\\
% \mathrm{zones~21-30:}&\alpha_3=0,~~\alpha_1,~\alpha_2,~\alpha_4\in\{\mathbb{R}_-,0,\mathbb{R}_+\},~\\
% \mathrm{zones~31-40:}&\alpha_4=0,~~\alpha_1,~\alpha_2,~\alpha_3\in\{\mathbb{R}_-,0,\mathbb{R}_+\},~\\
% \mathrm{zones~41-55:}&\alpha_4=0,~~\alpha_1,~\alpha_2,~\alpha_3\in\{\mathbb{R}_-,\mathbb{R}_+\}.\\
% \end{split}
% \]
\begin{figure}
\center%{\epsfig{figure=fig1.eps,width=6.85in}}
\includegraphics[width=3 in]{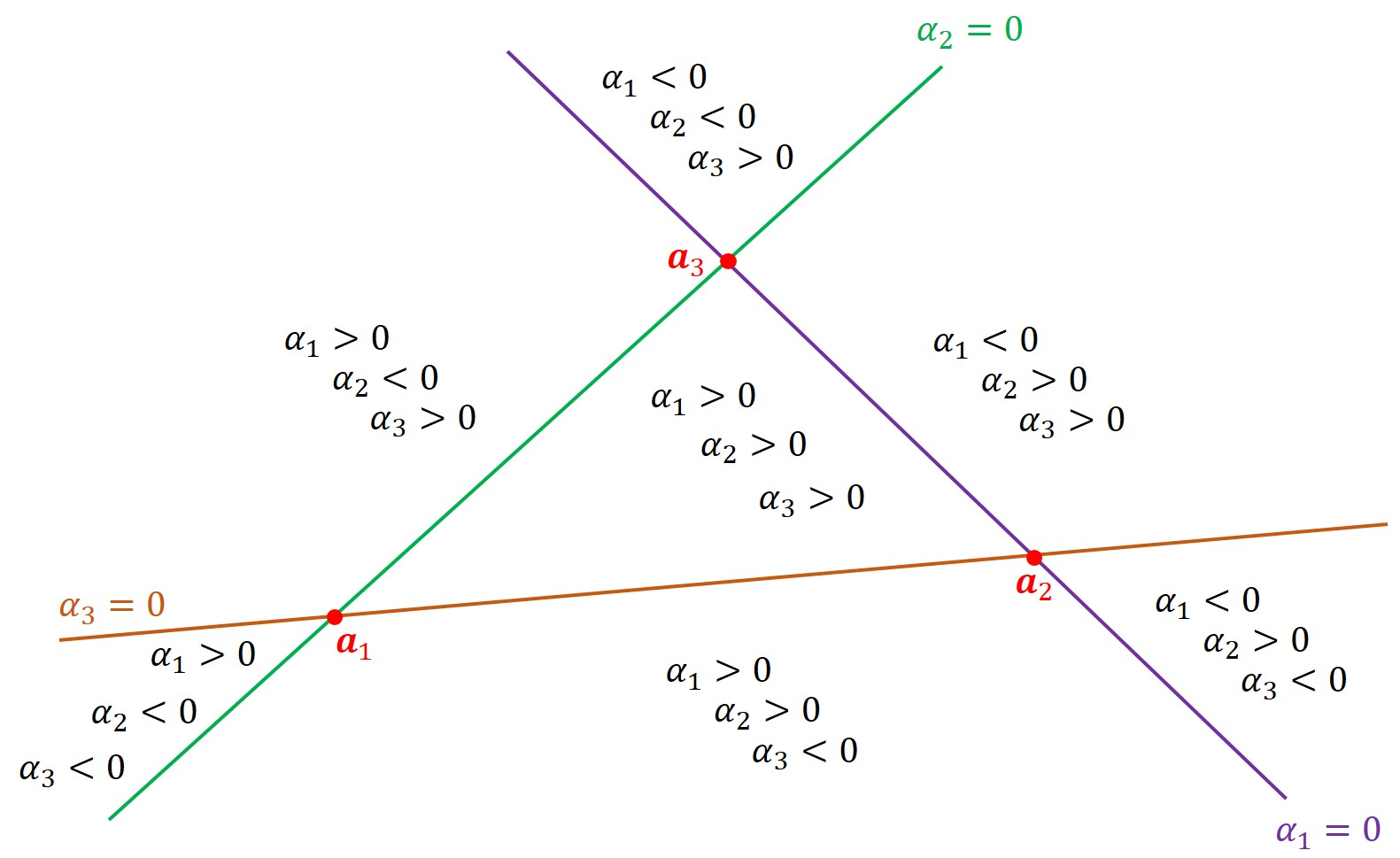}
\caption{2-dimension motion space democratization from $\alpha$ parameters.}
\label{2Ddisc}
\vspace{-20pt}
\end{figure}

\begin{proposition}
The point $\mathbf{c}$ is positioned \underline{inside} the (open) simplex defined by $\mathbf{a}_1$, $\cdots$, $\mathbf{a}_{\dim+1}$ if and only if  $\mathbf{\Omega}\left(\mathbf{a}_1,\cdots,\mathbf{a}_{\dim+1},\mathbf{c}\right)> \mathbf{0}$. 
% }
%
%  If $\neg\left(\mathbf{\Omega}\left(\mathbf{a}_1,\cdots,\mathbf{a}_{\dim+1},\mathbf{c}\right)> \mathbf{0}\right)$, then, $\mathbf{c}$ is positioned \underline{outside} the simplex defined by $\mathbf{a}_1$, $\cdots$, $\mathbf{a}_{\dim+1}$.
\end{proposition}

We use the term ``containment'' when a point $\mathbf{c}$ is inside a $d$-polytope, which typically represents a simplex of leaders, the motion space or an obstacle.
%simplex defined by $\mathbf{a}_1$, $\mathbf{a}_2$, $\cdots$, $\mathbf{a}_{n+1}$, e.g. $\mathbf{c}$ is contained by the simplex defined by $\mathbf{a}_1$, $\mathbf{a}_2$, $\cdots$, $\mathbf{a}_{n+1}$, if $\mathbf{\Omega}\left(\mathbf{a}_1,\cdots,\mathbf{a}_{\dim+1},\mathbf{c}\right)\geq \mathbf{0}$.

\subsection{Continuum Deformation Definition}
Consider an MAS consisting of $N$ agents identified by unique index numbers $\mathcal{V}=\{1,\cdots,N\}$.   Agents $1$ through $\dim+1$ are leaders and the remaining agents are followers acquiring the desired coordination through local communication, e.g. $\mathcal{V}_L=\{1,\cdots,\dim+1\}$ is the set of leaders and $\mathcal{V}_F=\{\dim+2,\cdots,N\}$ is the set of followers. 
We denote by ${\mathbf{r}}_i\left(t\right)$ the actual position of agent $i$ at time $t$, and by $\htr{\mathbf{r}}_i\left(t\right)$ its desired position at time $t$.
The $j$-th coordinate of ${\mathbf{r}}_i$ is denoted as $r_{i,j}$, and the $j$-th coordinate of $\htr{\mathbf{r}}_i$ is denoted as $\htr r_{i,j}$.
% The paper considers MAS evolution over time interval $[t_0,t_f]$, where initial time $t_0$  and final time $t_f$ are fixed and given. 
Let $\mathbf{r}_{i}^0$ and $\mathbf{r}_{i}^f$ denote initial and final positions of agent $i\in \mathcal{V}$, respectively.
% Desired position of agent $i$, denoted by $\htr{\mathbf{r}}_i\in \mathbb{R}^{\dim\times 1}$ at $t\in [t_0,t_f]$, where 
The desired position of agent $i$  is defined by:
\begin{equation}
\label{homogtransform}
    % t\in [t_0,t_f],\qquad 
    \htr{\mathbf{r}}_i\left(t\right)=\mathbf{Q}\left(t,t_0\right)\mathbf{r}_{i}^0+\mathbf{d}\left(t,t_0\right),
\end{equation}
where $\mathbf{r}_{i}^0=\htr{\mathbf{r}}_{i}\left(t_0\right)$, ${\htr{\mathbf{r}}_{i}}^f=\htr{\mathbf{r}}_{i}\left(t_f\right)$ ($i\in \mathcal{V}$), $t_0$ and $t_f$ denote initial and final time,
$\mathbf{Q}\left(t,t_0\right)\in \mathbb{R}^{\dim\times \dim}$ is the Jacobian matrix, $\mathbf{Q}\left(t_0,t_0\right)=\mathbf{I}_\dim\in \mathbb{R}^{\dim\times \dim}$ is the identity matrix,  $\mathbf{d}\left(t,t_0\right)\in \mathbb{R}^{\dim\times 1}$ is the rigid-body displacement vector, and $\mathbf{d}\left(t_0,t_0\right)=\mathbf{0}\in \mathbb{R}^{\dim\times 1}$. \change{The affine transformation \eqref{homogtransform} is called  \emph{homogeneous transformation} in continuum mechanics \cite{lai2009introduction}.} 

\change{In a homogeneous transformation coordination, leaders form \change{a} $\dim$-dimensional \emph{leading polytope} at any time $t$, therefore}
\begin{equation}
\label{leadersrankcondition}
\forall t,\qquad \mathbf{\Lambda}\left(\htr{\mathbf{r}}_{1},\cdots,\htr{\mathbf{r}}_{\dim+1} \right)=\dim.
\end{equation}
% 
% \begin{equation}
% \label{leadersrankcondition}
% \forall t,\qquad  \mathrm{Rank}\left(
% \begin{bmatrix}
% \htr{\mathbf{r}}_{2}-\htr{\mathbf{r}}_{1}&\cdots&\htr{\mathbf{r}}_{\dim+1}-\htr{\mathbf{r}}_{1}
% \end{bmatrix}
% \right)
% =\dim. 
% \end{equation}
\change{Because homogeneous transformation is a linear mapping}, $\mathbf{Q}$ and $\mathbf{D}$ elements are uniquely related to leader position components by
\begin{equation}
\label{QD}
\forall t,\qquad 
\begin{bmatrix}
 \mathrm{vec}\left(\mathbf{Q}^T\right)\\
    \mathbf{d}
\end{bmatrix}
=
\begin{bmatrix}
   \mathbf{I}_{\dim}\otimes \mathbf{P}(t_0)&\mathbf{I}_\dim\otimes \mathbf{1}_{\dim\times 1}
\end{bmatrix}
\mathrm{vec}\left(\mathbf{P}(t)\right),
\end{equation}
where "$\otimes$" is the Kronecker product, $\mathbf{1}_\dim\in \mathbb{R}^{\left(\dim+1\right)\times 1}$ is the one-entry matrix, and
\[
\mathbf{P}(t)=
\begin{bmatrix}
r_{1,1}^{\mathrm{HT}}&\cdots&r_{1,d}^{\mathrm{HT}}\\
\vdots&\vdots&\vdots\\
r_{d+1,d}^{\mathrm{HT}}&\cdots&r_{d+1,d}^{\mathrm{HT}}\\
\end{bmatrix}
\in \mathbb{R}^{\left(\dim+1\right)\times \dim}.
% ,~
% \mathbf{P}_{HT}=
% \begin{bmatrix}
% \htr{\mathbf{r}}_{1}^T\\
% \vdots\\
% \htr{\mathbf{r}}_{\dim+1}^T\\
% \end{bmatrix}
% \in \mathbb{R}^{\left(\dim+1\right)\dim\times 1}.
\]
%\textbf{Key Feature of Homogeneous Transformation:} 
%Under a homogeneous transformation,
$\mathbf{\Omega}\left(\htr{\mathbf{r}}_{1}\left(t\right),\cdots,\htr{\mathbf{r}}_{\dim+1}\left(t\right),\htr{\mathbf{r}}_i\left(t\right)\right)\in \mathbb{R}^{\left(\dim+1\right)\times 1}$ remains time-invariant at any time $t\in [t_0,t_f]$:
\begin{equation}
\forall t\in [t_0,t_f],\forall i\in \mathcal{V},\qquad  \mathbf{\Omega}\left(\htr{\mathbf{r}}_{1},\cdots,\htr{\mathbf{r}}_{\dim+1},\htr{\mathbf{r}}_i\right)=\mathbf{\Omega}_{i,0},
\end{equation}
is time-invariant, where
\[
    \forall i\in \mathcal{V},\qquad \mathbf{\Omega}_{i,0}=\mathbf{\Omega}\left(\mathbf{r}_{1}^0,\cdots,\change{\mathbf{r}_{\dim+1}^{\change{0}},\mathbf{r}_{i}^{\change{0}}}\right)\in \change{\mathbb{R}^{\dim+1}}.
\]
% Define 
% \[
% \begin{split}
% \mathbf{P}_{L,HT}\left(t\right)=&
% \begin{bmatrix}
% \htr{\mathbf{r}}_{1}&\cdots&\htr{\mathbf{r}}_{\dim+1}
% \end{bmatrix}
% \in \mathbb{R}^{\dim\times\left(N-\dim-1\right)},\\
% \mathbf{P}_{F,HT}\left(t\right)=&
% \begin{bmatrix}
% \htr{\mathbf{r}}_{\dim+2}&\cdots&\htr{\mathbf{r}}_{N}
% \end{bmatrix}
% \in \mathbb{R}^{\dim\times\left(N-\dim-1\right)},\\
% \mathbf{W}_L=&\begin{bmatrix}
% \mathbf{\Omega}^T\left(\mathbf{\mathbf{r}}_{1}^0,\cdots,\mathbf{r}_{\dim+1}^0,\mathbf{r}_{\dim+2}^0\right)\\
% \mathbf{\Omega}^T\left(\mathbf{\mathbf{r}}_{1}^0,\cdots,\mathbf{r}_{\dim+1}^0,\mathbf{r}_{\dim+3}^0\right)\\
% \vdots\\
% \mathbf{\Omega}^T\left(\mathbf{\mathbf{r}}_{1}^0,\cdots,\mathbf{r}_{\dim+1}^0,\mathbf{r}_{N}^0\right)\\
% \end{bmatrix}
% .
% \end{split}
% \]
% % \begin{theorem}
% Let $\mathbf{z}_{F,HT}=\mathrm{vec}\left({\mathbf{P}}_{F,HT}^T\right)\in \mathbb{R}^{\dim\left(N-\dim-1\right)\times 1}$ and $\mathbf{z}_{L,HT}=\mathrm{vec}\left({\mathbf{P}}_{F,HT}^T\right)\in \mathbb{R}^{\dim\left(N-\dim-1\right)\times 1}$. Then,
% \begin{equation}
%     t\in [t_0,t_f],\qquad \mathbf{z}_{F,HT}=\left(\mathbf{I}_\dim\otimes \mathbf{W}_L\right)\mathbf{z}_{L,HT}
% \end{equation}
% where $\mathbf{z}_{F,HT}$ defines follower desired configurations.

% If $\mathbf{Q}\in \mathbb{R}^{\dim\times \dim}$ is nonsingular and $\mathbf{Q}\left(t,t_0\right)=\mathbf{I}_\dim$, then, matrix $\mathbf{U}_D=\left(\mathbf{Q}^T\mathbf{Q}\right)$ is positive definite.
% \end{theorem}
% \begin{proof}
% See the proof in \cite{rastgoftar2016continuum}.
% \end{proof}
\noindent
\textbf{Assumption:} This paper assumes follower agents are positioned inside the leading simplex at initial time $t_0$: 
\[
\forall i\in \mathcal{V}_F,\qquad  \mathbf{\Omega}_{i,0}>0.
\]
\subsection{Continuum Deformation Acquisition}
% \subsubsection{Continuum Deformation Definition}
Assume directed graph $\mathcal{G}=\mathcal{G}\left(\mathcal{V},\mathcal{E}\right)$ defines \change{a fixed inter-agent communication topology}, $\mathcal{V}$ is the node set and $\mathcal{E}\subset \mathcal{V}\times \mathcal{V}$ is the edge set. Follower $i\in \mathcal{V}_F$ communicates with $\dim+1$ in-neighbor agents defined by
set $\mathcal{N}_i=\{i_1,\cdots,i_{\dim+1}\}\subset \mathcal{V}$. 
It is assumed that $\mathbf{\Lambda}\left(\mathbf{r}_{i_1,0},\cdots,\mathbf{r}_{i_{\dim+1},0}\right)=\dim$ ($\forall i\in \mathcal{V}_F$), so in-neighbor agents of follower $i$ form an $\dim$-dimensional simplex at initial time $t_0$. Follower inter-agent communications are weighted and obtained from
\begin{equation}
\label{ComWeights}
    \begin{bmatrix}
    w_{i,i_1}&\cdots&w_{i,i_{\dim+1}}
    \end{bmatrix}
    ^T=\mathbf{\Omega}\left(\mathbf{r}_{i_1}^0,\cdots,\mathbf{r}_{i_{\dim+1}}^0,\mathbf{r}_{i}^0\right).
\end{equation}
\change{Note that $w_{i,i_k}$ is the communication weight between follower $i\in \mathcal{V}_F$ and in-neghbpor agent $i_k\in \mathcal{N}_i$ ($k=1,\cdots,d+1$).}
% We define weight matrix $\mathbf{W}\in as: \mathbb{R}^{\left(N-\dim-1\right)\times N}$
% \begin{equation}
%     \mathbf{W}_{ij}=
%     \begin{cases}
%     w_{i,j}&i\in \mathcal{V}_F,~j\in \mathcal{N}_i\\
%     -1&j=i+\dim+1\\
%     0&\mathrm{else}.
%     \end{cases}
% \end{equation}
% $\mathbf{W}$ can be partitioned as follows:
% \begin{equation}
%     \mathbf{W}=
%     \begin{bmatrix}
%     \mathbf{B}&\mathbf{A}
%     \end{bmatrix}
%     .
% \end{equation}
% We prove $\mathbf{A}\in \mathbb{R}^{\left(N-\dim-1\right)\times \left(N-\dim-1\right)}$ is Hurwitz \cite{rastgoftar2016continuum} and define:

% \begin{equation}
% \mathbf{W}_L=-\mathbf{A}^{-1}\mathbf{B}=
% \begin{bmatrix}
% \mathbf{\Omega}^T\left(\mathbf{\mathbf{r}}_{1}^0,\cdots,\mathbf{r}_{\dim+1}^0,\mathbf{r}_{\dim+2}^0\right)\\
% \mathbf{\Omega}^T\left(\mathbf{\mathbf{r}}_{1}^0,\cdots,\mathbf{r}_{\dim+1}^0,\mathbf{r}_{\dim+3}^0\right)\\
% \vdots\\
% \mathbf{\Omega}^T\left(\mathbf{\mathbf{r}}_{1,0},\cdots,\mathbf{r}_{\dim+1}^0,\mathbf{r}_{N}^0\right)\\
% \end{bmatrix}
% .
% \end{equation}

\subsection{MAS Collective Dynamics Model}
\label{MAS Collective Dynamics Model} 
Let $\mathbf{r}_i\in \mathbb{R}^{\dim\times 1}$ denote actual position of agent $i\in \mathcal{V}$.
\begin{equation}
    \dfrac{\mathrm{d}^2\mathbf{r}_i}{\mathrm{d}t^2}=\mathbf{u}_i,
\end{equation}
where
%TODOjb
\begin{equation}
\begin{split}
    &\mathbf{u}_i=
    \begin{cases}
    \htr{\ddot{\mathbf{r}}}_i\left(\mathrm{given}\right)&i\in \mathcal{V}_L\\
    \beta_v\sum_{j\in \mathcal{N}_i}w_{i,j}\left(\dot{\mathbf{r}}_j-\dot{\mathbf{r}}_i\right)+\beta_r\sum_{j\in \mathcal{N}_i}w_{i,j}\left({\mathbf{r}}_j-{\mathbf{r}}_i\right)&i\in \mathcal{V}_F.\\
    \end{cases}
\end{split}
\end{equation}
% already defined.  and $w_{i,j}$ is communication weight between follower $i\in \mathcal{V}_F$ and  in-neighbor $j\in \mathcal{N}_i$, e.g. set $\mathcal{N}_i$ defines in-neighbor agents of follower $i$. 
For continuum deformation communication weights are consistent with agents' positions at $t_0$ and assigned by Eq. \eqref{ComWeights}.

\subsection{Temporal Logic}
Temporal Logic (TL) can capture temporal behavior of a dynamical system.
In this paper we use a logic based on LTL${}_{-X}$~\cite{DBLP:journals/tac/KloetzerB08}. The logic LTL${}_{-X}$ is a flavour of Linear Temporal Logic without the Next operator X (sometimes written $\circ$), which makes it more adapted to reasoning about continuous-time systems.
Since we are reasoning about an explicit system, we make our atomic formulas concrete, as comparisons of expressions.
Our logic uses two syntactic categories: expressions $e$ and propositions $\phi$. 
An expression $e$ can be a constant $c$, a state variable representing the $j$-th coordinate of the actual position of agent $j$, ${r}_{i,j}$, a state variable representing the $j$-th coordinate of the desired position of agent $j$, $\htr{{r}}_{i,j}$, as well as a multiplication $e_1\times e_2$, addition $e_1+e_2$, subtraction $e_1-e_2$, or division $e_1/e_2$ of two expressions. A formula can be True $\top$, a comparison of two expressions $e_1\leq e_2$, or a disjunction $\phi_1\lor\phi_2$, negation $\lnot\phi$ or until $\phi_1\mathcal{U}\phi_2$ of two formulas.
\begin{align*}
    e & ::=\ c\ |\ {{r}}_{i,j}\ |\ \htr{{r}}_{i,j}\ |\ e\times e\ |\ e+e\ |\ e-e\ |\ e/e\\
    \phi & ::=\ \top\ |\ e\leq e\ |\ \phi\lor\phi\ |\ \lnot\phi\ |\ \phi\mathcal{U}\phi
\end{align*}
We call \emph{atomic formulas} the formulas of the form $e\leq e$.
As is usual in LTL, we define the operators False $\bot$, conjunction $\land$, %implication $\Rightarrow$,
always $\Box$ and eventually $\Diamond$ as:
\begin{align*}
 \bot &= \lnot{\top}       & \Diamond\phi &= \top\mathcal{U}\phi \\
 \phi_1\land\phi_2 & = \lnot(\lnot \phi_1\lor\lnot\phi_2)  & \Box\phi & =  \lnot\Diamond\lnot\phi
 %\phi_1\Rightarrow\phi_2 & = (\lnot\phi_1)\lor\phi_2\\
\end{align*}
For any time $t\geq 0$, the state $\mathcal{S}(t)$ of our system is a function giving the valuation of every state variable: $\mathcal{S}(t):\{{r}_{1,1},\ldots,{{r}}_{N,\dim},\htr{{r}}_{1,1},\ldots,\htr{{r}}_{N,\dim}\}\rightarrow\mathbb{R}$
%We simply write ${\mathbf{r}}_{i,j}(t)$ for $\mathcal{S}(t)({\mathbf{r}}_{i,j})$ and similarly for $\htr{\mathbf{r}}_{i,j}(t)$ for $\mathcal{S}(t)(\htr{\mathbf{r}}_{i,j})$.
Given such a state $\mathcal{S}(t)$ for the valuation of state variables, an expression $e$ can be evaluated in the usual way to a real number that we write $\mathcal{S}(t)(e)$. The satisfaction of formula $\phi$ in state $\mathcal{S}(t)$ (i.e., at time $t$) is then given by:
\begin{align*}
    S(t)&\vDash \top \text{ is always satisfied;}\\
    S(t)&\vDash e_1\leq e_2 \text{ if and only if } \mathcal{S}(t)(e_1)\leq\mathcal{S}(t)(e_2);\\
    S(t)&\vDash \lnot\phi \text{ if and only if }S(t)\not\vDash\phi;\\
    S(t)&\vDash \phi_1\lor\phi_2 \text{ if and only if}S(t)\vDash\phi_1\text{ or }S(t)\vDash\phi_2;\\
    S(t)&\vDash \phi_1\mathcal{U}\phi_2 \text{ if and only if there exists $t'\geq t$ such that}\\
    & \text{$S(t')\vDash \phi_2$ and for all $t\leq t''< t'$ we have $S(t'')\vDash \phi_1$.}
\end{align*}

For convenience, we write $e^2$ for the expression $e\times e$; $\|\mathbf{r}_i-\htr{\mathbf{r}}_i\|_2^2$ for the expression 
$({r}_{i,1}-{r}_{i,1})^2+\cdots+({r}_{i,\dim}-\htr{{r}}_{i,\dim})^2$;
and $\mathbf{\Omega}\left(\htr{\mathbf{r}}_{1},\cdots,\htr{\mathbf{r}}_{\dim+1},\mathbf{r}_i\right)$ as in Equation~\ref{eq:omega} (Section~\ref{sec:triangulation}).

\section{Formal Specification}
\label{problemstatement}
This paper's first objective is to formally specify safety requirements for continuum deformation. MAS continuum deformation is considered safe if the following requirements are satisfied: (1)~Bounded deviation, (2)~Follower containment guarantee, (3)~Inter-agent collision avoidance, (4)~Motion-space containment, and (5)~Obstacle collision avoidance.

The paper's second objective is to formally specify a liveness condition: agent desired final position reachability. 
% Necessary conditions of this condition are also  specified.

% following three liveness requirements: (6) Existence of the leading polytope,  (7) Validity of the MAS desired destination given MAS reference formation, obstacle and motion space geometries, and controller performance. (8) Eventual arrival to the MAS desired destination. 

\textbf{Definition 1 (Motion Space):} The motion space, denoted by $\mathbf{B}\subset \mathbb{R}^\dim$, is finite and convex. Let $\mathbf{B}$ enclose $m_B$ \change{simplexes} $\mathbf{B}_1$, $\cdots$, $\mathbf{B}_{m_B}$, e.g. $\bigcup_{i=1}^{m_B}\mathbf{B}_i\subset \mathbf{B}$. $\mathbf{B}_i$ is a \change{$\dim$-dimensional} simplex with vertices at $\mathbf{b}_{i,1}\in \mathbb{R}^{\dim\times 1}$, $\cdots$ $\mathbf{b}_{i,\dim+1}\in\mathbb{R}^{\dim\times 1}$.

\textbf{Definition 2 (Obstacle):} Let $\mathcal{O}\subset \mathbb{R}^\dim$ be a finite set defining motion space obstacles. Let $\mathcal{O}$ encompass $m_O$ simplex es $\mathbf{O}_1$, $\cdots$, $\mathbf{O}_{m_O}$, e.g. $\mathcal{O}\subset \bigcup_{i=1}^{m_O}\mathbf{O}_i$. $\mathbf{O}_i$ is an $\dim$-dimensional simplex with vertices $\mathbf{o}_{i,1}\in \mathbb{R}^{\dim\times 1}$, $\cdots$ $\mathbf{o}_{i,\dim+1}\in \mathbb{R}^{\dim\times 1}$.

% \subsection{MAS Continuum Deformation Safety Requirements}
% \label{MAS Continuum Deformation Safety Requirements}

\renewcommand\theparagraphdis{\arabic{paragraph})}
%\jb{TODO: change the names of the conditions to something with indices 1, 2, 3..., 8. I think it will be more readable}

\paragraph{Safety Condition 1: Bounded Vehicle Deviation}
\noindent \change{Deviation of every agent from continuum deformation must not exceed $\delta$, i.e., the actual position $\mathbf{r}_i$ ($i\in \mathcal{V}$) of every agent must stay within $\delta$ of its desired position $\mathbf{r}^{HT}_i$.}  This requirement can be expressed as:
\begin{equation}
\label{deltaineq}
    % \forall i\in \mathcal{V},\qquad 
    \fbox{$\displaystyle\bigwedge_{i\in \mathcal{V}}\Box\left(\|\mathbf{r}_i-\htr{\mathbf{r}}_i\|_2^2\leq \delta^2\right),\tag{$\psi_1$}$}
\end{equation}
where $\delta$ is constant and $\|\cdot\|_2$ is the 2-norm symbol. 

\paragraph{Safety Condition 2: Follower Containment Condition}
\noindent Follower $i\in \mathcal{V}_F$ must be inside the leading simplex at any time $t$. This condition can be expressed as:
\begin{equation*}
%\label{Followerrcontainnment0}
    \forall i\in \mathcal{V}_F,\forall t\geq t_0\qquad 
    \mathbf{r}_i\in
    \mathcal{T}(\htr{\mathbf{r}}_{1},\cdots,\htr{\mathbf{r}}_{\dim+1})
\end{equation*}
which can be expressed in our logic using the function $\mathbf{\Omega}$ as:
\begin{equation}
\label{Followerrcontainnment}
    % \forall i\in \mathcal{V}_F,\qquad 
    \fbox{$\displaystyle\bigwedge_{i\in \mathcal{V}_F}\Box\left(\mathbf{\Omega}\left(\htr{\mathbf{r}}_{1},\cdots,\htr{\mathbf{r}}_{\dim+1},\mathbf{r}_i\right)\geq 0\right).\tag{$\psi_2$}$}
\end{equation}

\paragraph{Safety Condition 3: Inter-Agent Collision Avoidance}
\label{Inter-Agent Collision Avoidance and UAV Containment}
\noindent Assume every agent is enclosed by a ball of radius $\epsilon$. Collision avoidance between any two different agents $i$ and $j$ is satisfied, if \change{and only if:}
\begin{equation}
\label{MicColAvoid}
  \fbox{$\displaystyle\bigwedge_{i,j\in \mathcal{V},~
    i\neq j}\Box\left(\|\mathbf{r}_{i}-\mathbf{r}_j\|_2^2\geq (2\epsilon)^2\right).$} \tag{$\psi_3$}
\end{equation}

\paragraph{Safety Condition 4: Motion Space Containment}
\noindent Motion space containment is satisfied, if
\[
\forall i\in \mathcal{V},\forall t\geq t_0\qquad \mathbf{r}_i\in \mathbf{B}
\]
which can be expressed in our logic using the function $\mathbf{\Omega}$ as:
\begin{equation}\label{motionspacecontainment}
\begin{split}
% \forall i\in \mathcal{V},\qquad 
\fbox{$\displaystyle\bigwedge_{i\in \mathcal{V}}\Box\bigvee_{k=1}^{m_B}\left(\mathbf{\Omega}\left(\mathbf{b}_{k,1},\cdots,\mathbf{b}_{k,\dim+1},\mathbf{r}_i\right)\geq 0\right).$}\\
    % &\vdots\\
    % \exists k_{\dim+1}\in\{1,\cdots,m_B\}, \Psi_{\dim+1,3}:\bigcap_{i\in \mathcal{V}_L}&\mathbf{\Omega}\left(\mathbf{b}_{k_{\dim+1},1},\cdots,\mathbf{b}_{k_{\dim+1},m_B},\htr{\mathbf{r}}_{\dim+1}\right)&\geq 0,\\
\end{split}
\tag{$\psi_4$}
\end{equation}
Eq. \eqref{motionspacecontainment} ensures existence of a simplex $\mathbf{B}_i\subset \mathbf{B}$ enclosing leader $i\in \mathcal{V}_L$ at any time $t$.

\paragraph{Safety Condition 5: Obstacle Collision Avoidance}
\noindent Obstacle collision avoidance is satisfied if
\[
\forall i\in \mathcal{V},\forall t\geq t_0,\qquad \Box(\mathbf{r}_i\notin \mathcal{O}).
\]
which can be expressed in our logic using the function $\mathbf{\Omega}$ as: % Obstacle collision is avoided if the following two conditions are satisfied:
\begin{subequations}
\begin{equation}\label{obstacleavoidance}
\begin{split}
    % &\forall i\in\mathcal{V},~~
    \fbox{$\displaystyle\bigwedge_{i\in \mathcal{V}}\Box\left(\bigwedge_{k=1}^{m_O}\neg\left(\mathbf{\Omega}\left(\mathbf{o}_{k,1},\cdots,\mathbf{o}_{k,m_B},\mathbf{r}_i\right)\geq 0\right)\right).$}
\end{split}
\tag{$\psi_5$}
\end{equation}
%\begin{equation}
%\label{LeadingPolytopeAvoidance}
%    k\in \{1,\cdots,m_O\},\qquad \Box \bigwedge_{k %=1}^{\dim+1}\neg\left(\mathbf{\Omega}\left(\htr{\mathbf{r}}_%{1},\cdots,\htr{\mathbf{r}}_{\dim+1},\mathbf{o}_{k,i}\%right)\geq 0\right).\tag{$\theta_k$}
%\end{equation}
\end{subequations}
Eq. \eqref{obstacleavoidance} ensures every agent $i\in \mathcal{V}$ is outside the obstacle zone defined by simplexes $\mathbf{O}_1$, $\cdots$, $\mathbf{O}_{m_O}$. 
%Eq. \eqref{LeadingPolytopeAvoidance} ensures that all simplex es $\mathbf{O}_1$, $\cdots$, $\mathbf{O}_{n_O}$ are outside the leading polytope.

\paragraph{Liveness Condition 6: Final Formation Rechability}
\noindent Given agent desired final positions $\mathbf{r}_{1}^f$, $\cdots$, $\mathbf{r}_{N}^f$, the liveness condition is defined by:
\begin{equation}
    % \forall i\in \mathcal{V},\qquad 
    \fbox{$\displaystyle\Diamond\Box\bigwedge_{i\in \mathcal{V}}\left(\|\mathbf{r}_i- \mathbf{r}_{i}^f\|_2^2\leq\varepsilon^2\right)\tag{$\psi_6$}.$}
\end{equation}

 \section{Sufficient Conditions}\label{Formal Specification}
 \subsection{Inter-Agent Collision Avoidance and Agent Containment}
 It is computationally expensive to ensure inter-agent collision avoidance and follower containment using Eqs. \eqref{MicColAvoid} and \eqref{Followerrcontainnment}. We can instead use the sufficient conditions provided in Theorem \ref{theorm1} to guarantee these two MAS safety constraints at less computational cost.   
\begin{theorem}\label{theorm1}\cite{rastgoftar2016continuum}
 Let $D_B$ denote minimum separation distance between two agents at  initial time $t_0$, and let $D_S$ denote the minimum boundary distance at initial time $t_0$. Define
\[
\delta_{\mathrm{max}}=\min\left\{{1\over 2}\left(D_B-2\epsilon\right),\left(D_S-\epsilon\right)\right\}
\]
and
\begin{equation}
\label{Lmin}
\lambda_{\mathrm{min}}=\dfrac{\delta+\epsilon}{\delta_{\mathrm{max}}+\epsilon}.    
\end{equation}
Inter-agent collision avoidance and agent containment \change{are} guaranteed, if the
eigenvalues of pure deformation matrix $\mathbf{U}_D=\left(\mathbf{Q}^T\mathbf{Q}\right)^{1\over 2}$, denoted $\lambda_1$, $\lambda_2$, and $\lambda_3$, satisfy 
\begin{equation}
\label{Interagentcollisionavoidance}
\forall t\geq0,\qquad   \bigwedge_{i=1}^3 \left(\lambda_{\mathrm{min}}\leq \big|\lambda_i\left(t\right)\big|\right),
\end{equation}
and no agent deviation exceeds $\delta$ at any time $t$. 
% In other words,
% \[
% \left(\wedge_{i=1}^3\Psi_{i,5}\right)\wedge\left(\wedge_{\forall i\in \mathcal{V}}\Psi_{i,1}\right)\rightarrow \wedge_{\forall i\in \mathcal{V}}\Box\left(\Psi_{i,1}\wedge\Psi_{i,4}\right)
% \]
\end{theorem}
\begin{proof}\cite{rastgoftar2016continuum}
Let $m_1$ and $m_2$ denote two points of the leading simplex that has the minimum separation distance at $t_0$. If $\delta_{\mathrm{max}}={1\over 2}\left(D_B-\epsilon\right)$ then $m_1, m_2\in \mathcal{V}$ are two agents (Fig. \ref{InirFormDsDb}(c)). Otherwise, $m_1\in \mathcal{V}_F$ is the index number of a follower and $m_2$ denotes a point on the boundary of the leading simplex having minimum distance from $m_1$ (Fig. \ref{InirFormDsDb}(b)):
\[
\|\mathbf{r}_{m_1}^0-\mathbf{r}_{m_2}^0\|_2=\mu\left(\delta_{\mathrm{max}}+\epsilon\right),
\]
where 
\[
\mu=
\begin{cases}
2& m_1,m_2\in \mathcal{V}\\
1&m_1\in \mathcal{V}_F,~m_2~\mathrm{is~at~the~leading~polytope~boundary.}
\end{cases}
\]
Considering Eq. \eqref{homogtransform},
\[
\begin{split}
\left(\mathbf{r}_{m_2}-\mathbf{r}_{m_1}\right)^T\left(\mathbf{r}_{m_2}-\mathbf{r}_{m_1}\right)=&\left(\mathbf{r}_{m_2}^0-\mathbf{r}_{m_1}^0\right)^T\mathbf{U}_D^2\left(\mathbf{r}_{m_2}^0-\mathbf{r}_{m_1}^0\right).
% \\
% \leq &
\end{split}
\]
Assume 
\[
\forall i,j\in \mathcal{V},i\neq j,\qquad \Box\left(\left(\delta+\epsilon\right)\leq \min\|\mathbf{r}_i-\mathbf{r}_j\|_2\right),
\]
then, inter-agent collision avoidance is ensured if inequality \eqref{deltaineq} is satisfied. This implies that
\[
\begin{split}
\mu^2\left(\delta+\epsilon\right)^2\leq& \min \big\{\lambda_1^2,\lambda_2^2,\lambda_3^2\big\}\mu^2\left(\delta_{\mathrm{max}}+\epsilon\right)^2\\
\leq &\left(\mathbf{r}_{m_2,0}-\mathbf{r}_{m_1,0}\right)^T\mathbf{U}_D^2\left(\mathbf{r}_{m_2,0}-\mathbf{r}_{m_1,0}\right).
\end{split}
\]
In other words, inter-agent collision avoidance is avoided if
\[
\forall t,~i=1,2,3,\qquad \left(\left(\dfrac{\delta+\epsilon}{\delta_{\mathrm{max}}+\epsilon}\right)^2\leq \lambda_i^2\left(t\right)\right).
\]
Consequently,  inter-agent collision is avoided if inequality \eqref{Interagentcollisionavoidance} is satisfied. Because $\mathbf{Q}$ is nonsingular at any time $t$ and $\mathbf{Q}(t_0,t_0)=\mathbf{I}_d$, $\mathbf{U}_D$ eigenvalues are always positive. Therefore, Eq. \eqref{Interagentcollisionavoidance} is satisfied.

\end{proof}

\begin{figure*}[!ht]
 \centering
  \subfigure[]{\includegraphics[width=0.28\linewidth]{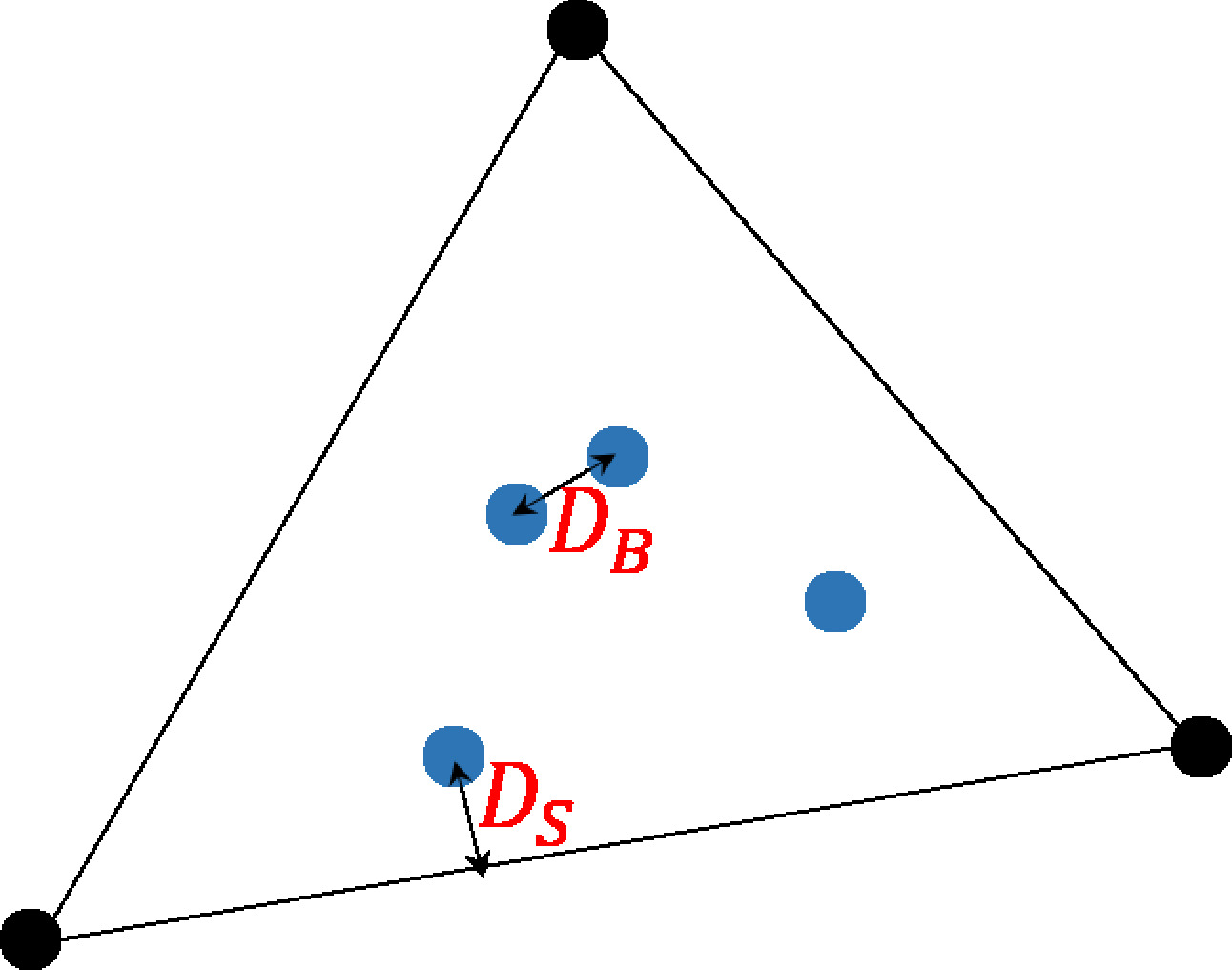}}
 \subfigure[$\mu=1$]{\includegraphics[width=0.27\linewidth]{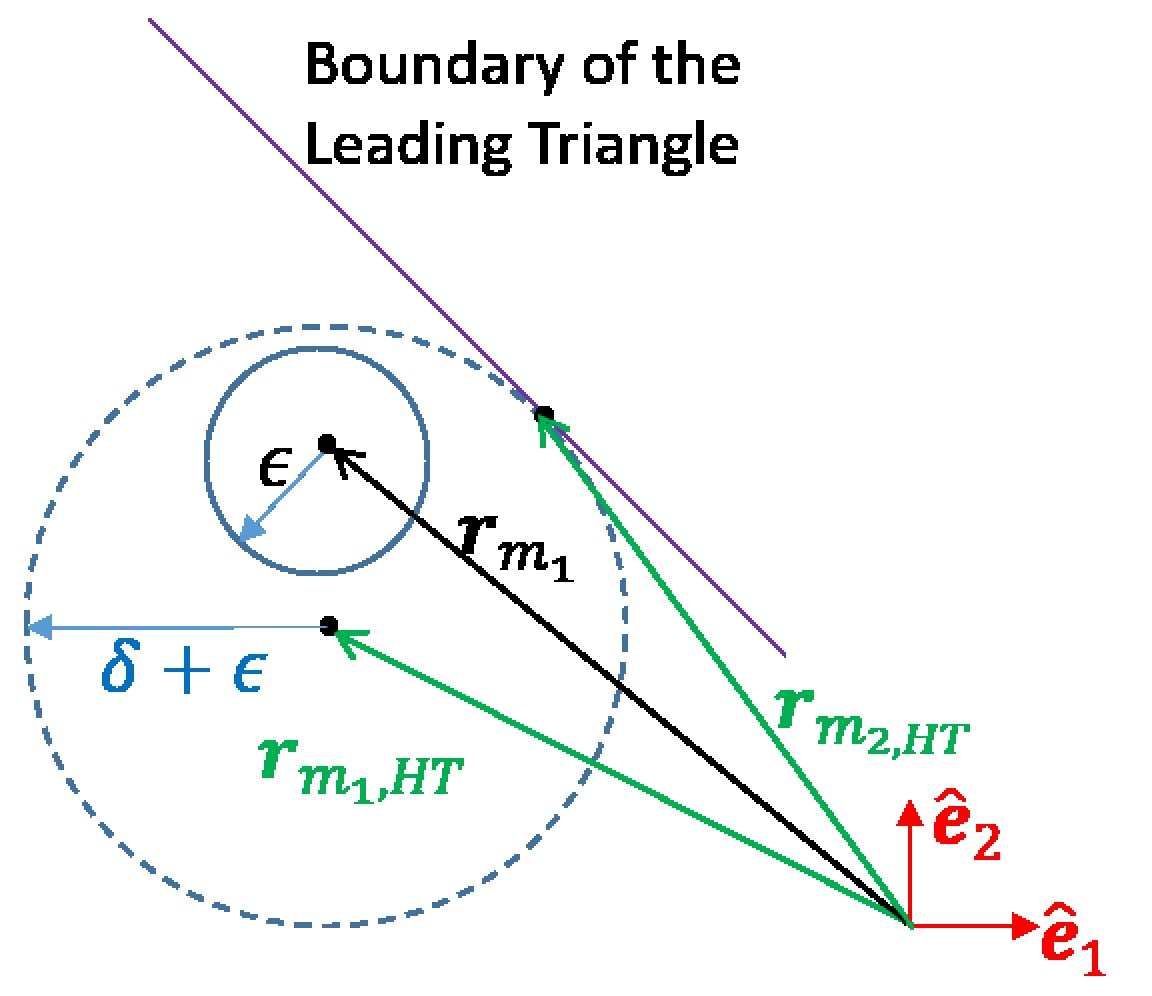}}
\subfigure[$\mu=2$]{\includegraphics[width=0.25\linewidth]{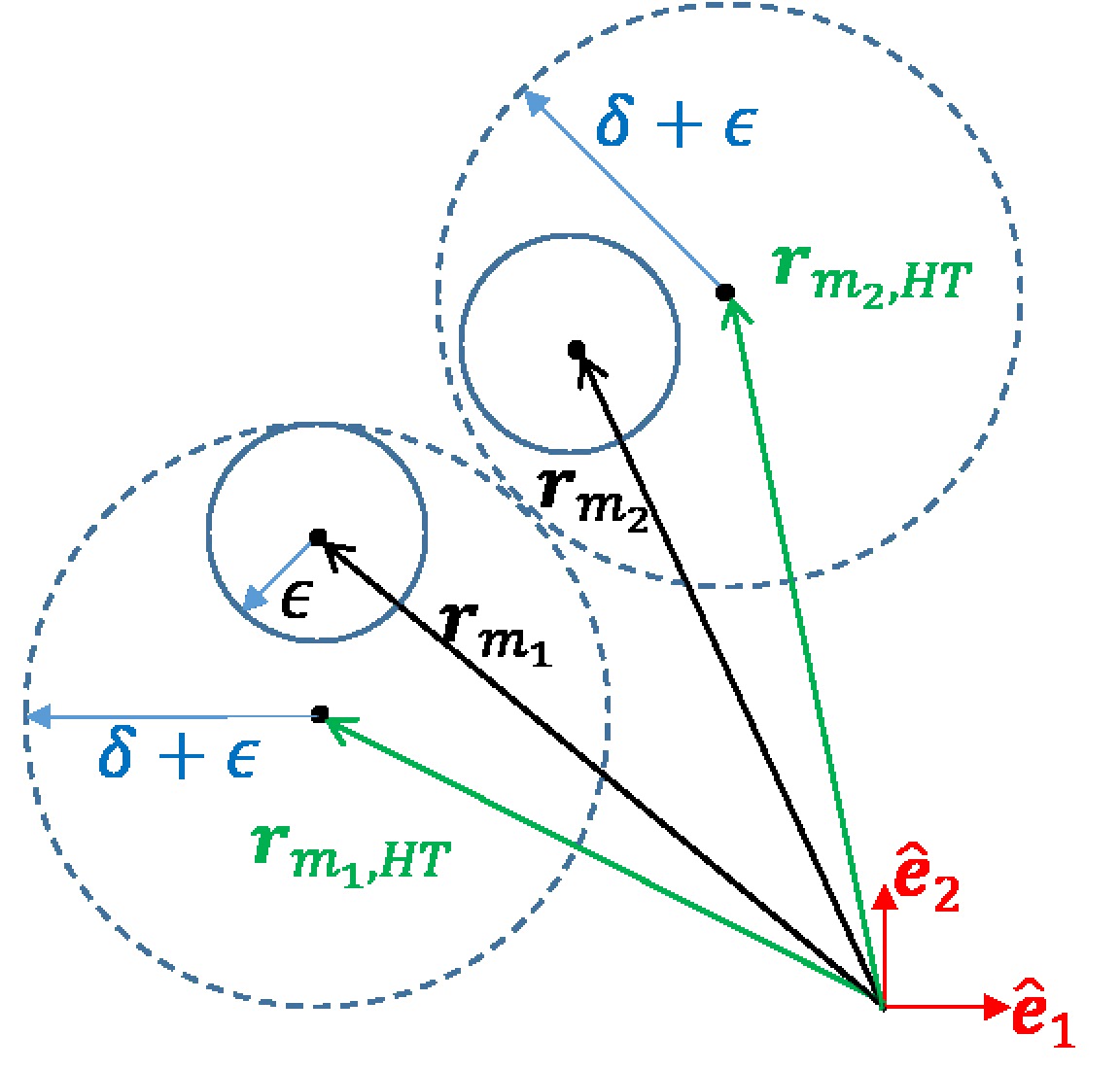}}
 \caption{(a) Minimum distances $D_B$ and $D_S$ at $t_0$. (b)  $ D_S-\epsilon<0.5\left(D_B-2\epsilon\right)$ ($\mu=1$), $\delta_{\mathrm{max}}$ is assigned based on the closest distance from the boundary. (c)  $ 0.5\left(D_B-2\epsilon\right)\leq D_S-\epsilon$ ($\mu=2$), $\delta_{\mathrm{max}}$ is assigned based on agents $m_1$ and $m_2$ having the closest separation distance at $t_0$. $\mathbf{r}_{m_1}$ and $\mathbf{r}_{m_2}$ are the actual positions of points $m_1$ and $m_2$.}
 \vspace{-15pt}
\label{InirFormDsDb}
\end{figure*}

\begin{figure}
\center%{\epsfig{figure=fig1.eps,width=6.85in}}
\includegraphics[width=3in]{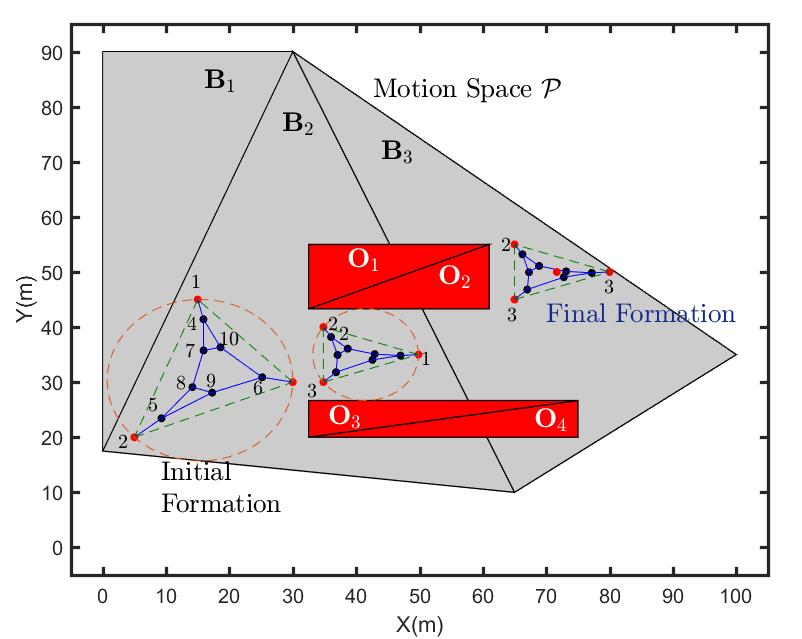}
\caption{Schematic of motion space $\mathbf{B}$.}
\label{initfindorm}
\end{figure}

% \noindent \textbf{Assumption 1:} At $t_0$, the leading simplex is outside  $\mathcal O$:
% \[
% \mathbf{c}\in \mathbf{B},\qquad \left(\Omega\left(\mathbf{r}_{1,0},\cdots,\mathbf{r}_{\dim+1,0},\mathbf{c}\right)\geq 0\right)\cap\mathcal{O}=\emptyset.
% \]
% \textbf{Assumption 2:} At time $t_0$, leading simplex is inside $\mathbf{B}$:
% \[
% \mathbf{c}\in \mathbf{B},\qquad \left(\Omega\left(\mathbf{r}_{1,0},\cdots,\mathbf{r}_{\dim+1,0},\mathbf{c}\right)\geq 0\right)\cup\mathbf{B}=\mathbf{B}.
% \]
\subsection{Motion Space Containment}
If safety condition $\psi_2$ is satisfied, then motion space containment is guaranteed by ensuring leaders remain inside the motion space $\mathbf{B}$. Formally, given the formula:
\begin{equation}
\bigwedge_{i\in \mathcal{V}_L}\Box\bigvee_{k=1}^{m_B}\left(\mathbf{\Omega}\left(\mathbf{b}_{k,1},\cdots,\mathbf{b}_{k,\dim+1},\mathbf{r}_i\right)\geq 0\right),\tag{$\psi_{7}$}
\end{equation}
we have:
\begin{theorem}
 If $\psi_2\land\psi_7$ is satisfied, then $\psi_4$ is satisfied.
\end{theorem}
\subsection{Obstacle Collision Avoidance}
If safety condition $\psi_2$ is satisfied, then obstacle collision avoidance is guaranteed by ensuring leaders  do not collide with obstacles. Formally, given the formula:
\begin{equation}
\bigwedge_{i\in \mathcal{V}_L}\Box\left(\bigwedge_{k=1}^{m_O}\neg\left(\mathbf{\Omega}\left(\mathbf{o}_{k,1},\cdots,\mathbf{o}_{k,m_B},\mathbf{r}_i\right)\geq 0\right)\right),\tag{$\psi_{8}$}
\end{equation}
we have:
\begin{theorem}
 If $\psi_2\land\psi_{8}$ is satisfied, then $\psi_5$ is satisfied. 
\end{theorem}
Proofs of Theorems 2 and 3 are adapted from \cite{rastgoftar2016continuum}.
\section{Simulation Results}
\label{Simulation Results}
Consider \change{an} MAS with $N=10$ agents evolving in 2 dimensions ($\dim=2$). Agents $1$, $2$, and $3$ are leaders; the remaining agents are followers. Inter-agent communication is defined by the Fig.~\ref{initfindorm} graph, and follower communication weights are listed in Table \ref{Table2}. Follower communication weights are consistent with the initial formation and assigned by Eq.~\eqref{ComWeights}.
  Fig. \ref{initfindorm} also shows MAS initial and final formations.  $\mathbf{B}=\mathbf{B}_1\bigcup\mathbf{B}_2\bigcup\mathbf{B}_3$ defines the motion space, and $\mathcal{O}=\bigcup_{k=1}^4\mathbf{O}_4$ defines obstacles in $\mathbf{B}$.
% \begin{figure}
% \center%{\epsfig{figure=fig1.eps,width=6.85in}}
% \includegraphics[width=3.0 in]{comunication2D.eps}
% \caption{Communication graph used by followers to acquire the prescribed homogeneous deformation}
% \label{fig3}
% \end{figure}
\begin{table}[t]
\caption{Communication weights $w _{i,i_1}$, $w _{i,i_2}$, and $w _{i,i_3}$}
\label{Table2}
\begin{center}
\begin{tabular}{c l l l l l l }
\hline
$i$ & $i_1$&$i_2$&$i_3$& $w _{i,{i_1}}$&$w_{i,{i_2}}$&$w _{i,{i_{3}}}$ \\
\hline
$4$&$1$& $7$&$10$&	$0.60$&	 $0.20$&	$0.20$\\
$5$&$2$&	$8$&	$9$&	$0.60$&	    $0.20$&	    $0.20$\\
$6$&	$3$&	$9$&	$10$&	$0.60$&	    $0.20$&	    $0.20$\\
$7$&	$4$& $8$ &$10$&	$0.40$&	    $0.36$&	    $0.24$\\
$8$&$5$& $7$ &$9$&	${1\over3}$&	    ${1\over3}$&	    ${1\over3}$\\
$9$&$5$ &$6$& $8$&	$0.31$&	    $0.42$&	    $0.27$\\
$10$&	$4$ &$6$ &$7$&	$0.35$&	    $0.29$&	    $0.36$\\
\hline
\end{tabular}
\end{center}
\end{table}
%\vspace{-0.6cm}
% The communication weight matrix is defined as:
% % \\
% % \\
% \[
% \setlength\arraycolsep{.9pt}
% \begin{split}
% \mathbf{W}=
% \begin{bmatrix}
% \mathbf{B}\big|\mathbf{A}
% \end{bmatrix}
% =
% \left[
% \begin{array}{ccc|ccccccc}
% 0.6&0&0&-1&0&0&0.2&0&0&0.2\\
% 0&0.6&0&0&-1&0&0&0.2&0.2&0\\
% 0&0&0.6&0&0&-1&0&0&0.2&0.2\\
% 0&0&0&0.4&0&0&-1&0.36&0&0.24\\
% 0&0&0&0&{1\over 3}&0&{1\over 3}&-1&{1\over 3}&0\\
% 0&0&0&0&0.31&0.42&0&0.27&-1&0\\
% 0&0&0&0.35&0&0.29&0.36&0&0&-1\\
% \end{array}
% .
% \right]
% \end{split}
% \]
%  Eigenvalues of $\mathbf{A}$  are $-0.2407$, $-0.5456$, $-0.7941$, $-1.1422$, $-1.2541$, $-1.5228$, and $-1.5006$. Therefore, $\mathbf{A}$ is Hurwitz and MAS evolution dynamics \eqref{MASEVOLDYN} is stable. 
 The paper assumes all agents are identical with $\beta_r=2$ and $\beta_v=4$. Agent positions are plotted versus time in Figs. \ref{PostionComponents} (a) and  \ref{PostionComponents} (b) with $t\in [0,227.5],~t_0=0s,~t_f=227.5s$.

% To incorporate the requirements 2 and 3

% To incorporate the second requirement, we assign bounds on 
% \[
% J_i=
% \begin{bmatrix}
% \dfrac{\partial \sigma_{1,i}}{\partial \beta_{1,i}}&\dfrac{\partial \sigma_{1,i}}{\partial \beta_{2,i}}\\
% \dfrac{\partial \sigma_{2,i}}{\partial \beta_{1,i}}&\dfrac{\partial \sigma_{2,i}}{\partial \beta_{2,i}}\\
% \end{bmatrix}
% =
% \begin{bmatrix}
% 1&0\\
% 0&1\\
% \end{bmatrix}
% .
% \]
% Motion filed is defined by
% \[
% \Omega_P=\bigg\{(x,y)|4.5\leq x\leq9.5,1.5\leq y\leq16.5\bigg\}
% \]
% Actual position of the agent $i\in V$ is denoted by 
% \[
% \mathbf{r}_i=x_{i}\hat{\mathbf{e}}_1+y_{i}\hat{\mathbf{e}}_2+=50+{1\over 4}x_{i}\hat{\mathbf{e}}_3
% \]
% % where $z(x_{i},y_{i})=50+{1\over 2}x_{i}$.

 \begin{figure}
\centering
\subfigure[]{\includegraphics[width=0.8\linewidth]{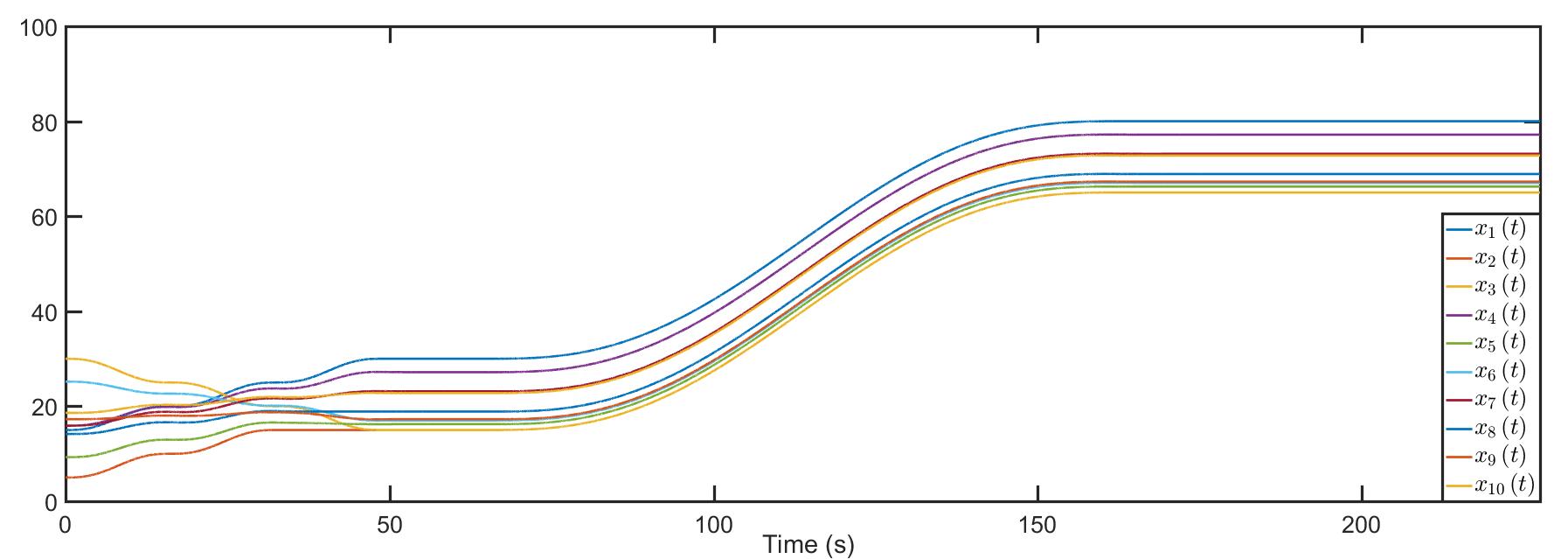}}
\subfigure[]{\includegraphics[width=0.8\linewidth]{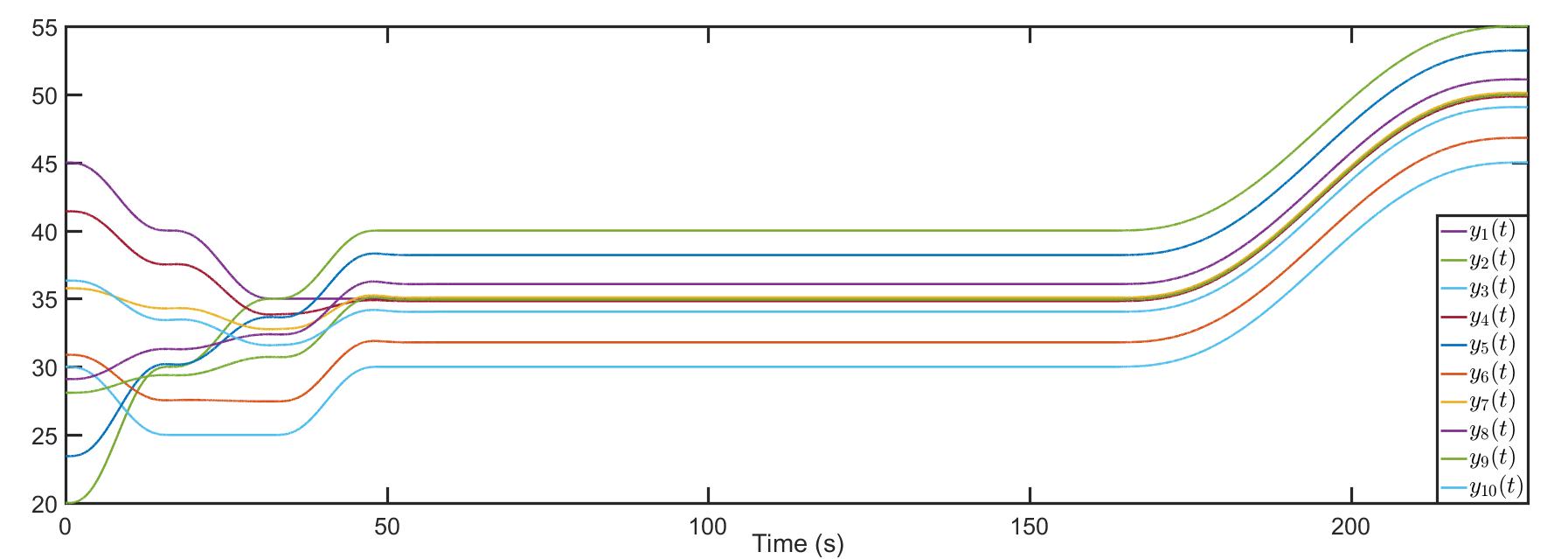}}
\subfigure[]{\includegraphics[width=0.8\linewidth]{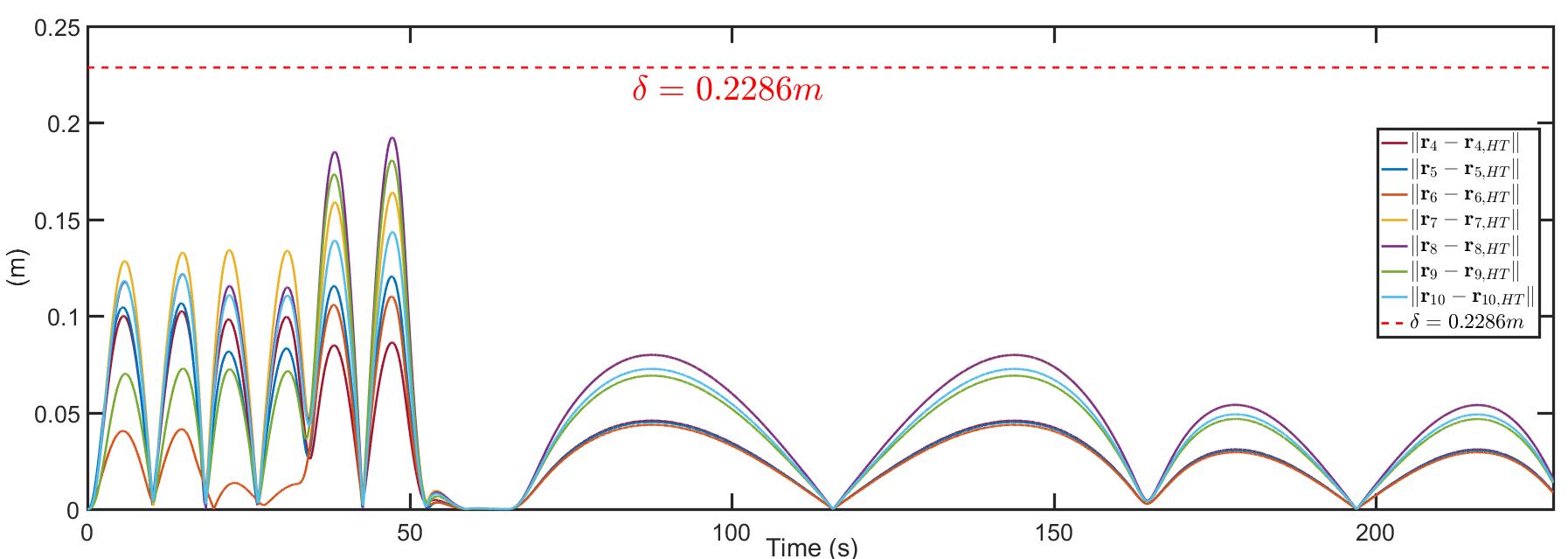}}
\caption{(a,b) $x$ and $y$ components of agents' actual positions versus time; (c) Deviation of follower agents versus time.}
\label{PostionComponents}
\end{figure}
\begin{figure}
\center%{\epsfig{figure=fig1.eps,width=6.85in}}
\includegraphics[width=2.8 in]{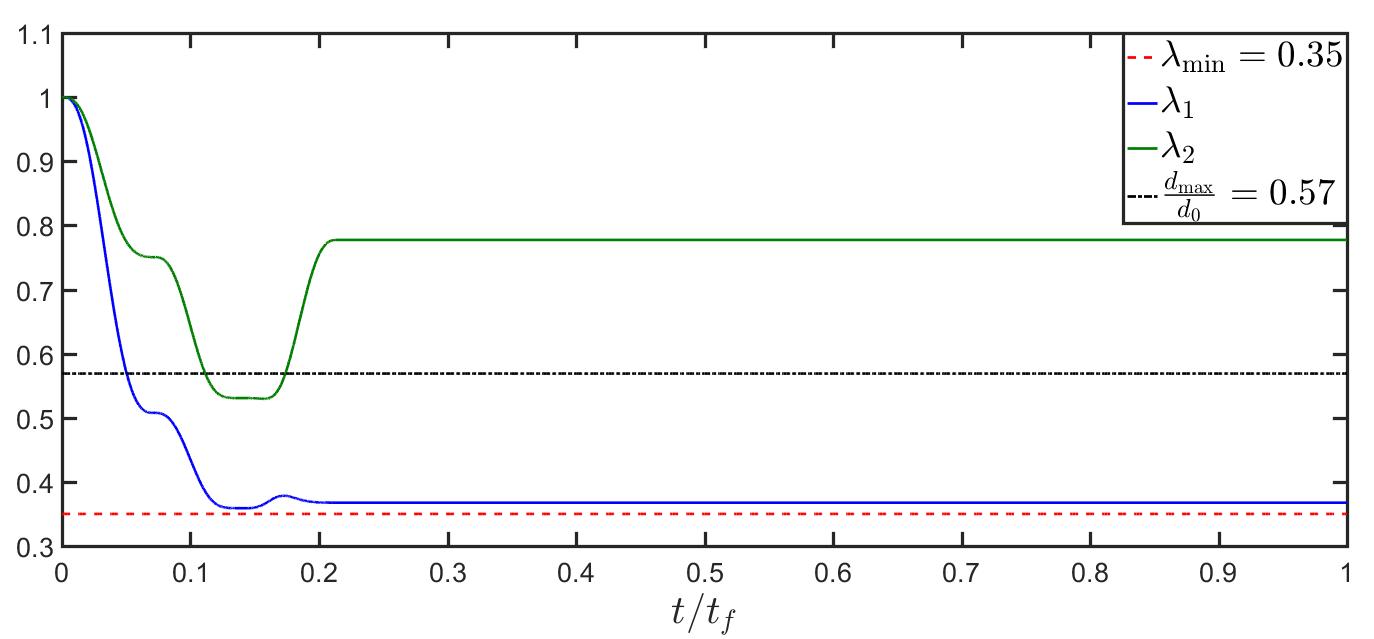}
\caption{Eigenvalues of the matrix $U_D$ versus time}
\label{LAMBDATHETA}
\end{figure}

\textbf{Satisfaction of Safety Condition 1:} Fig. \ref{PostionComponents}(c) plots deviation of every follower versus time confirming that no follower exceeds $\delta=0.2286m$ at any time $t\in[t_0,t_f]s$.

\textbf{Satisfaction of Safety Conditions 2 and 3:} Given MAS initial formation, $D_B=2.7348m$ and $D_S=1.5996m$ are the minimum separation and boundary distances. The paper assumes that each agent is enclosed by a ball with radius $\epsilon=0.25m$, thus, $\delta_{\mathrm{max}}$ $\delta_{\mathrm{ma}}= 1.1174m$. Given $\delta=0.2286$, $\epsilon=0.25m$, and $\delta_{\mathrm{ma}}= 1.1174m$,  $\lambda_{\mathrm{min}}=0.35$ is computed by Eq. \eqref{Lmin}. As shown in Fig. \ref{LAMBDATHETA},  $\mathbf{U}_D$ eigenvalues are greater than $\lambda_{\mathrm{min}}$ at any time $t$, hence, safety condition 2 is satisfied.

\textbf{Satisfaction of Safety Conditions 4 and 5:} Leader paths are plotted in Figs. \ref{leaders123} (a-c). As shown, motion containment and obstacle collision avoidance conditions are satisfied.

\begin{figure*}[!ht]
 \centering
  \subfigure[]{\includegraphics[width=0.32\linewidth]{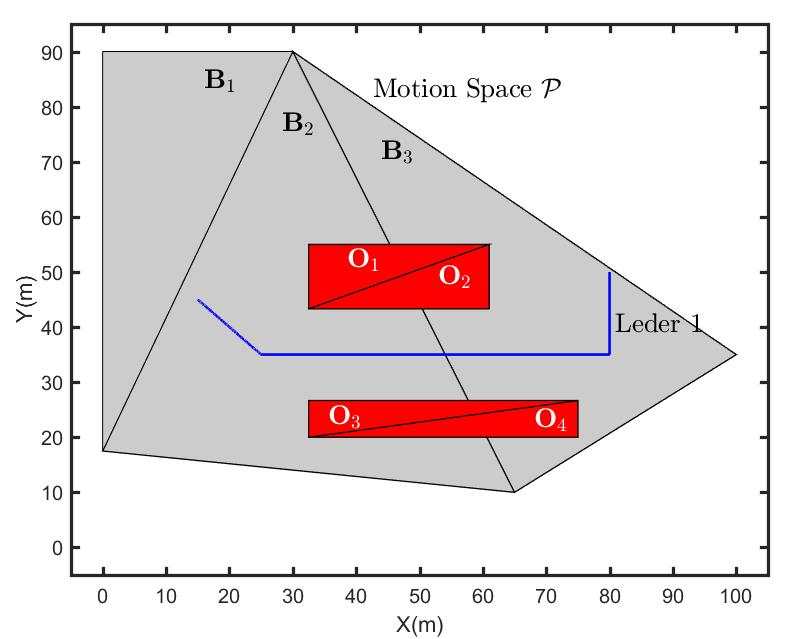}}
 \subfigure[]{\includegraphics[width=0.32\linewidth]{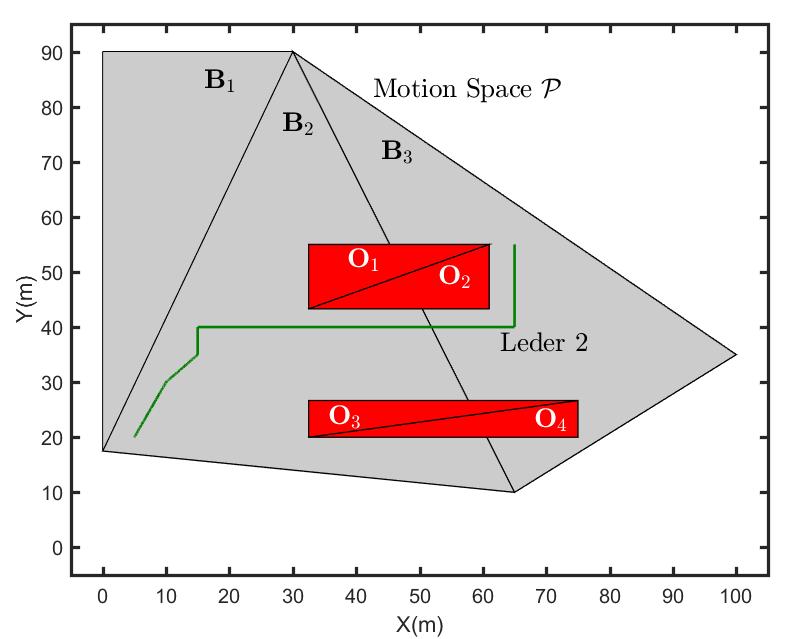}}
\subfigure[]{\includegraphics[width=0.32\linewidth]{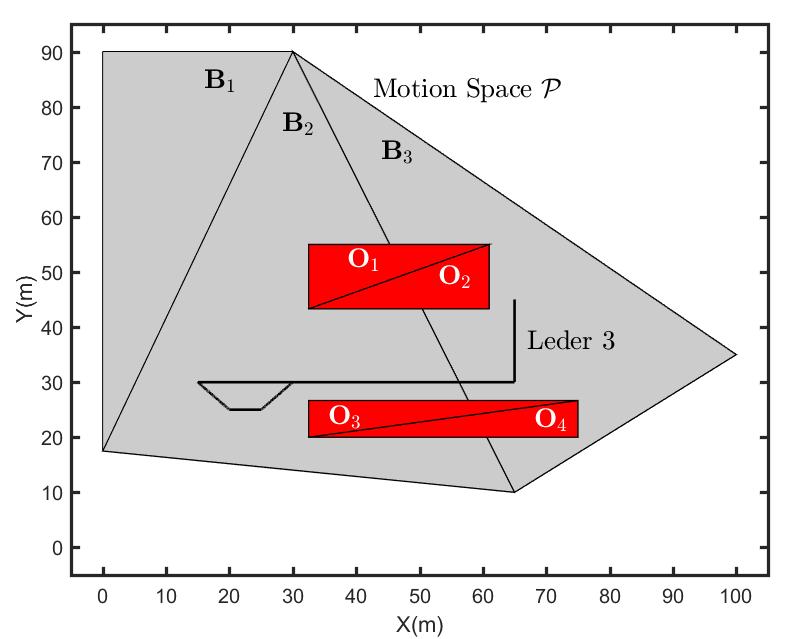}}
 \caption{Paths of the continuum deformation leaders: (a) Leader $1$, (b) Leader $2$, (c) Leader $3$.}
  \vspace{-15pt}
\label{leaders123}
\end{figure*}

\textbf{Satisfaction of Necessary Condition 6:} 
% The initial containment ball has diameter $d_0=14.6472m$. Given obstacle geometry, $d_{\mathrm{max}}=8.3333m$. Therefore, ${d_{\mathrm{max}}\over d_0}>\lambda_{\mathrm{min}}$ and liveness condition \eqref{livenesssss} is satisfied (Fig. \ref{LAMBDATHETA}). The initial containment ball and the largest ball inside the narrow passage are shown by dashed orange circles in Fig. \ref{initfindorm}. As shown in Fig. \ref{LAMBDATHETA}, $\lambda_{1}^f=\lambda_1\left(t_f\right)=0.3676$ and $\lambda_{2}^f=\lambda_2\left(t_f\right)=0.7772$. Because the $\mathbf{U}_D$ eigenvalues are positive at final time $t_f=227.5s$, necessary condition 2 is satisfied. 
As shown in Fig. \ref{LAMBDATHETA}, $\|\mathbf{r}_i-\mathbf{r}_{i}^{\mathrm{HT}}\|$ tends to zero at final time $t_f$, therefore, the liveness condition 6 is satisfied.

% \textbf{Satisfaction of Liveness Condition 6:}

 \section*{Acknowledgements}
 This work was supported in part by National Science Foundation Grant CNS 1739525.

\section{Conclusion}
\label{Conclusion}
\change{
In this paper we formally specified continuum deformation coordination in a $\dim$-dimensional motion space. Using triangulation and tetrahedralization, we developed safety and liveness conditions for continuum deformation. We constructed Linear Temporal Logic (LTL) formulae to check the validity of inter-agent and obstacle collision avoidance as well as agent and motion-space containment. We demonstrated validity of the method with simulation results. The paper shows how a large-scale continuum deformation satisfies the liveness and safety conditions we developed. This formal definition supports efficient specification and computational overhead when designing and deploying a large-scale MAS.}

%%%%%%%%%%%%%%%%%%%%%%%%%%%%%%%%%%%%%%%%%%%%%%%%%%%%%%%%%%%%%%%%%%%%%%%%%%%%%%%%

\bibliographystyle{IEEEtran}
\bibliography{reference}

\end{document}